\documentclass[a4paper,reqno,12pt]{article}
\usepackage{amsthm,amsmath}    
\usepackage{amssymb}
\usepackage{amsfonts}

\advance\textheight by 1.0in

\theoremstyle{plain}

\newtheorem{theorem}{Theorem}[section]
\newtheorem{lemma}[theorem]{Lemma}
\newtheorem{corollary}[theorem]{Corollary}
\newtheorem{proposition}[theorem]{Proposition}

\newtheorem{definition}[theorem]{Definition}

\theoremstyle{remark}

\renewenvironment{proof}{\noindent{\it Proof}.}{\qed}

\numberwithin{equation}{section}

\title{Reducibility of Euler integrals and multiintegrals}
\author{ E.I. Ganzha\thanks{The author was supported by the Russian
Foundation for Basic Research under the grant No~06-01-00814 and a grant 09-09-1/NSh
from Krasnoyarsk State Pedagogical University.} \\
 Department of Mathematics, \\
   Krasnoyarsk State Pedagogical University,\\
ul. Lebedevoi, 89,\\
 660060 Krasnoyarsk \\
e-mail:\ \ \texttt{eiganzha@mail.ru} }

\date{}

\newcommand{\ord}{\ensuremath \mathop \mathrm{ord} \nolimits}

\begin{document}
\maketitle
\begin{abstract}

We discuss the notion of reduction of a special type of explicit
solutions which generalize the solutions appearing in the classical
Laplace cascade method of integration of hyperbolic equations of the
second order in the plane. We give algorithms of reduction and prove
that different natural precise definitions of reduction are
equivalent.


\end{abstract}

\section{Introduction}

Classical methods of integration of linear partial differential
equations and systems of such equations (cf.\
\cite{forsyth,Darboux2,Gu}) were considerably generalized in the
last decades in \cite{Zh-St,St08,Athorne,ts05}. New applications of
these methods were found in two-dimensional spectral theory
\cite{TT07} and the theory of stochastic systems \cite{glt07}. The
following expressions appear as examples of exact solutions in the
methods developed in the aforementioned publications:
\begin{equation}\label{000}
\begin{array}{rl}
   f(x,y)= &  a_0(x,y)\varphi (x) + a_1(x,y) \varphi'(x)+ \ldots + a_n(x,y)
\varphi^{(n)}(x) + {} \\[0.5em]
    & \ \ \ \ \ \ \ \ \ \ \ \ {}+ b_0(x,y)\psi (y) +  \ldots +
b_m(x,y)\psi^{(m)}(y)
\end{array}
\end{equation}
with arbitrary functions $\varphi (x)$, $\psi (y)$ and fixed
coefficients $a_i(x,y)$, $b_j(x,y)$. First examples of such
solutions were found already very early by Euler \cite{Eul-III}. For
systems (\cite{Zh-St,St08,Athorne,ts05}) and higher-order equations
(\cite{LeRoux,Petren,pisati}) solutions may appear as sums of
several expression of the form (\ref{000}). Already Darboux and
other authors of the classical period had noted that the expressions
of the form (\ref{000}) sometimes may be simplified (i.e.\ the
orders $n$ and $m$ can be made lower) if one introduces new
arbitrary functions, for example as $\vartheta(x)=r_0(x)\varphi (x)
+ \ldots + r_k(x) \varphi^{(k)}(x)$. No general theory of
simplification (reduction) of the expression of type (\ref{000}) is
currently available. In this paper we give algorithms of reduction
and prove that different natural precise definitions of reduction
are equivalent. For simplicity we separate the parts of the
expression (\ref{000}) containing $\varphi (x)$ and $\psi (y)$, and
give the following definition:
\begin{definition}\label{def-Euler}
Euler integral in the plane is an expression of the form
\begin{equation}\label{2.1}
    \mathcal{I}=a_0(x,y)\varphi (x) + a_1(x,y) \varphi'(x)+ \ldots + a_n(x,y)
    \varphi^{(n)}(x),
\end{equation}
where $\varphi (x)$ is an arbitrary function of the variable $x$ and
$a_i(x,y)$ are given functions of the two variables $(x,y)$.
\end{definition}

Below we assume that $a_i(x,y)$ belong to some constructive
differential field of functions in the plane. In order to guarantee
correctness of our algorithms and results we need to require that we
can constructively decide if a given element of this field is zero;
we also obviously need at least existence of derivatives of the
coefficients $a_i(x,y)$ up to the order needed in the computations.
All functions should have a common domain of definition---some open
subset of $\mathbb{R}^2$. The simplest practically important example
of such a field is the field of rational functions
$\mathbb{Q}(x,y)$.

Hereafter we will often call Euler integrals (\ref{2.1}) simply
``integrals'' and write them as $\mathcal{I}= L(x,y)\varphi (x)$,
where
$$L(x,y) =
a_0(x,y)  + a_1(x,y)D  + \ldots + a_n(x,y)D^n, \qquad
D=\frac{d}{dx},
$$
is a linear ordinary differential operator (LODO). All  differential
operators in this paper will be linear ordinary differential
operators in $x$, i.e. they will include the derivatives
$D^s=d^s/dx^s$ only. Dependence of the coefficients of an operator
on the both variables $x$ and $y$ or on $x$ alone will be explicitly
shown after the sign of an operator. Everywhere below we will denote
linear ordinary differential operators in $x$ with upper-case Latin
letters  and all given functions with lower-case Latin letters;
lower-case Greek letters will denote arbitrary functions on $x$.

\begin{definition}\label{def-multiint}
Euler multiintegral in the plane is an expression of the form
     \begin{equation}\label{3.1}
\begin{array}{rcl}
  \mathcal{J}  & = & a_{10}(x,y)\varphi_1(x) +  a_{11}(x,y)\varphi_1'(x) +
      \ldots + a_{1n_1}(x,y)\varphi_1^{(n_1)}(x) +{} \\[0.3em]
  & &  {} + a_{20}(x,y)\varphi_2(x) + a_{21}(x,y)\varphi_2'(x) +
      \ldots + a_{2n_2}(x,y)\varphi_2^{(n_2)}(x) +{} \\
    & &  \ \ \ \ \ \ldots  \\
  &  &  {} + a_{k0}(x,y)\varphi_k(x) +  a_{k1}(x,y)\varphi_k'(x) +
      \ldots + a_{kn_k}(x,y)\varphi_k^{(n_k)}(x) ={} \\[1em]
    & = & L_1(x,y)\varphi_1(x) + \ldots +  L_k(x,y)\varphi_k(x), \\
\end{array}
     \end{equation}
where $\varphi_1(x)$, \ldots , $\varphi_k(x)$ are arbitrary
functions of the variable $x$ alone and $a_{ij}(x,y)$ are given
functions of the two variables $(x,y)$, $1 \leq i \leq k$,  $1 \leq
j \leq n_i$.
\end{definition}

For the sake of brevity we will denote the operator row
$\big(L_1(x,y), \ldots , L_k(x,y)\big)$ as $\overline{L}(x,y)$, and
the column of functions $\big(\varphi_1(x) , \ldots,
\varphi_k(x)\big)^t$ as $\mathring{\varphi}(x)$:
\begin{equation}\label{3.2}
  \mathcal{J} =\overline{L}(x,y)\mathring{\varphi}(x).
\end{equation}

In Petr\'en's thesis~\cite{Petren} the following linear partial
differential operators in the plain were considered:
\begin{equation}\label{P.1.1}
P(x,y)=\sum_{i=0}^{p-1} a_i(x,y)D_y^i D_x +\sum_{i=0}^{p-1}
b_i(x,y)D_y^i = A(x,y) D_x + B(x,y); \quad a_{p-1} \neq 0.
\end{equation}
Now we will show that any Euler integral
\begin{equation}\label{P.1.2}
\mathcal{I}=L(x,y)\varphi(x)
\end{equation}
can be an (incomplete) solution of some equation of Petr\'en type
$P(x,y)\mathcal{I} = 0$.

Suppose we have an integral $\mathcal{I}$ generated by an operator
$L(x,y)=l_n(x,y)D_x^n +l_{n-1}(x,y)D_x^{n-1} + \ldots +l_0(x,y)$. We
have to find an operator $P(x,y)$ of the form (\ref{P.1.1}), such
that
\begin{equation}\label{P.1.3}
    \big(P(x,y)  L(x,y)\big)\,\varphi(x)\equiv 0
\end{equation}
for any $\varphi(x)$. Since
$$
\begin{array}{c}
P(x,y)
L(x,y)=A(x,y)D_xL(x,y)+B(x,y)L(x,y)=\\[0.3em]
=A(x,y)L(x,y)D_x+A(x,y)[D_x,L(x,y)]+B(x,y)L(x,y)=
\\[0.3em]
=A(x,y)\big[l_n(x,y)D_x^{n+1} +l_{n-1}(x,y)D_x^{n} + \ldots
+l_0(x,y)D_x\big] +
\\[0.5em]
+A(x,y)\big[(l_n)'_x(x,y)D_x^{n} +(l_{n-1})'_x(x,y)D_x^{n-1} +
\ldots +(l_0)'_x(x,y)\big] +
\\[0.5em]
+B(x,y)\big[l_n(x,y)D_x^{n} +l_{n-1}(x,y)D_x^{n-1} + \ldots
+l_0(x,y)\big]
\end{array}
$$
and $\varphi(x)$ is an arbitrary function, then  (\ref{P.1.3}) is
equivalent to the system
\begin{equation}\label{P.2.4}
\left\{
\begin{array}{l}
  A(x,y) l_n(x,y)=0,\\[0.3em]
  A(x,y)\big[l_n(x,y) +(l_{n-1})'_x(x,y)\big] +B(x,y) l_n(x,y)=0, \\[0.3em]
  \ \ \ \ \ \ \ldots \\
  A(x,y)\big[l_0(x,y) +(l_{1})'_x(x,y)\big] +B(x,y) l_1(x,y)=0,  \\[0.3em]
  A(x,y) (l_{0})'(x,y) + B(x,y) l_{0}(x,y) =0,
\end{array}
\right.
\end{equation}
where $l_i(x,y)$ are the given coefficients of the operator
$L(x,y)$. Since (\ref{P.2.4}) is a system with $(n+2)$ linear
homogeneous algebraic equations for $2p$ unknown coefficients
$b_{0}$, $b_{1}$, \ldots, $b_{p-1}$, $a_{0}$, \ldots , $a_{p-1}$,
one easily concludes that (\ref{P.2.4}) always has a nontrivial
solution if $p
> \frac{n+2}{2}$ and consequently every integral of the form
(\ref{P.1.2}) is a solution (certainly an incomplete solution in
general) of some equation of Petr\'en type~(\ref{P.1.1}) of
sufficiently high degree.

Analogously for multiintegrals~(\ref{3.1}), it is possible to find
some equation of Petr\'en type satisfied by a given multiintegral.
In this case the coefficients of the equation~(\ref{P.1.1}) satisfy
a set of systems of the form~(\ref{P.2.4}) (a system~(\ref{P.2.4})
is formed for every operator $L_i(x,y)$) with the total number of
linear homogeneous algebraic equations equal to $N=\sum_{i=1}^k n_i
+ 2k$, where $n_i$ is the order of the operator $L_i(x,y)$. So if
$2p > N$ there is a nontrivial solution of such a system for the
coefficients $a_{i}$, $b_{i}$ of an equation of Petr\'en type.

The linear case considered in this paper can be a basis for a
further study of exact solutions of nonlinear PDEs integrable by
Darboux method: as we know (see \cite{2ZhS2001,anderson,Gu}) Darboux
integrability is equivalent to integrability of the corresponding
linearized equation by the Laplace cascade method.

\section{Euler integrals}

In order to give the most general precise definition of reducibility
of a Euler integral (\ref{2.1}) we indroduce the following natural
definition:

\begin{definition}\label{def-zapas}
The set of all functions $Z(\mathcal{I})$ of the two variables
$(x,y)$ which will be obtained after substitution of arbitrary
(smooth) functions $\varphi (x)$ into (\ref{2.1}) will be called the
function stock generated by some Euler integral $\mathcal{I}$ of the
form (\ref{2.1}).
\end{definition}

Any such  function stock is an infinite-dimensional linear space;
its elements may be obviously added and multiplied with constants.

\begin{definition}\label{def-ord} For any Euler integral
(\ref{2.1}) we will call the order $n$ of the leading derivative of
the arbitrary function $\varphi (x)$ the order of the integral.
\end{definition}

\begin{definition}\label{def-red-ord} We will call a Euler integral
(\ref{2.1}) \emph{order-reducible} if there exists another integral
$\mathcal{I}_1 = b_0(x,y)\psi (x) + \ldots + b_m(x,y)\psi ^{(m)}(x)$
of smaller order $m$, $m<n$, generating the same function stock as
the given integral $\mathcal{I}$.
\end{definition}

Note that we do not assume any relation between the functions
$\varphi (x)$ and $\psi (x)$ in $\mathcal{I}$ and $\mathcal{I}_1$.
As we prove below, if an integral is order-reducible then there
exists a differential relation between $\varphi (x)$ and $\psi (x)$.

\begin{definition}\label{def-red-oper} We call a  Euler integral
(\ref{2.1}) \emph{operator-reducible} if the corresponding operator
$L(x,y)$ may be represented as a composition of operators
\begin{equation}\label{2.3}
    L(x,y) = M(x,y)  R(x),
\end{equation}
where $M(x,y)$ is another operator of lower order and the
coefficients of the operator $R(x)$ depend (as shown) only on $x$.
\end{definition}

In this case the function stocks generated by the integrals
$\mathcal{I} = L(x,y)\varphi (x)$ and $\mathcal{I}_1 = M(x,y)\psi
(x)$ obviously coincide and there exists the differential relation
$\psi (x) = R(x)\varphi (x)$.

\begin{definition}\label{def-red-ker} We will say that  a  Euler integral
(\ref{2.1}) has no kernel  (is \emph{kernel-irreducible}), if
$\mathcal{I} = L(x,y)\varphi (x)$ identically vanishes as a function
of the two variables only for $\varphi (x) \equiv 0$. In the
opposite case we will call the integral $\mathcal{I}$
\emph{kernel-reducible}.
\end{definition}

Below we will need the following constructive way of order reduction
for any Euler integral.

\newpage

\begin{center}
    \textbf{Algorithm} $\mathbf{RI}$
\end{center}

 The input data of the algorithm is a  Euler integral
(\ref{2.1}). The output of the algorithm is a representation of this
integral in the form (\ref{2.3}) with an operator $M(x,y)$ of
minimal possible order or a proof that the given integral is
operator-irreducible.

Let us take two new formal variables $y_1$ and $y_2$. Consider
$L(x,y_1)$ and $L(x,y_2)$ as two different ordinary operators in
$D_x$ with coefficients depending on the parameters $y_1$ and $y_2$
respectively. One can find the right greatest common divisor of
these operators $rGCD(L(x,y_1), L(x,y_2))$ using the well-known
non-commutative version of the Euclidean algorithm (cf. for example
\cite{Jacobson}), performing divisions of these linear ordinary
differential operators with remainders:
\begin{equation}\label{E1}
     \begin{array}{rcl}
       L(x,y_1) &=& A_0(x,y_1,y_2)L(x,y_2) + B_1(x,y_1,y_2), \\
       L(x,y_2) &=& A_1(x,y_1,y_2)B_1(x,y_1,y_2) + B_2(x,y_1,y_2), \\
       \ldots\ \ \  &\ldots & \ \ \ \ \ \ \ \ \  \ldots\\
       B_{k-1}(x,y_1,y_2) &=& A_k(x,y_1,y_2)B_k(x,y_1,y_2).
     \end{array}
\end{equation}
So we have $rGCD(L(x,y_1), L(x,y_2))=B_k(x,y_1,y_2)$.

\emph{Case 1: $\ord B_k(x,y_1,y_2)=0$.} Then the operators
$L(x,y_1)$ and $L(x,y_2)$ have no nontrivial  right greatest common
divisor and our algorithm terminates. The integral $\mathcal{I} =
L(x,y)\varphi (x)$ is operator-irreducible.

\emph{Case 2: $\ord B_k(x,y_1,y_2)=\ord L(x,y_1)=\ord L(x,y_2) =n$.}
This means that the Euclidean algorithm has terminated on the first
division step and $L(x,y_1) = a_0(x,y_1,y_2)L(x,y_2)$, $\ord a_0 =
0$. Consequently $L(x,y) = a_0(x,y,y_0)L(x,y_0)$, where $y_0$ is
some fixed generic point (such that $a_n(x,y_0) \not\equiv 0$ and
$a_0(x,y,y_0)$ is correctly defined), i.e.\ $L(x,y)$ is a LODO with
coefficients depending on $x$ only, multiplied on the left with a
function in the two variables $(x,y)$. Thus we have obtained a
representation (\ref{2.3}) with $\ord M(x,y)=0$ and $R(x)=L(x,y_0)$.
As a result we conclude that the integral $\mathcal{I}$ is
operator-reducible to another integral
$\mathcal{I}_1=a_0(x,y,y_0)\psi(x)$ of order zero.

\emph{Case 3: $\ord B_k(x,y_1,y_2) =m$, $0<m<\ord L(x,y)$.} Then, if
$B_k(x,y_1,y_2) =B_k(x)$, i.e.\ does not depend on $y_1$, $y_2$,
$R(x)=B_k(x)$ so we again obtain the required reduction of the
integral. In the opposite case we introduce a new formal variable
$y_3$ and return to the stage (\ref{E1}) of our algorithm
$\mathbf{RI}$ applying the Euclidean algorithm to the operators
$L(x,y_3)$ and $B_k(x,y_1,y_2)$. Again some operator
$B^{(1)}_r(x,y_1,y_2,y_3)=rGCD(L(x,y_1), L(x,y_2),L(x,y_3))$ will be
found, and in the Cases~1 and 2 (when $\ord B^{(1)}_r=0$ or $\ord
B^{(1)}_r=\ord B_k$) the algorithm terminates; in the Case~3 we
return to the stage (\ref{E1}) adding a new formal variable $y_4$
and finding $B^{(2)}_k(x,y_1,y_2,y_3,y_4)=rGCD(L(x,y_1), \ldots
,L(x,y_4))$. Since on every cycle for the Case~3 the order of
$B^{(i)}_k$ decreases at least by one, after a finite number of
steps the algorithm terminates. As a result the operators $M(x,y)$
and $R(x)$ appearing in the representation (\ref{2.3}) will be
found; obviously the order of $R(x)$ will be maximal and the
integral $\mathcal{I}_1 = M(x,y)\psi (x)$ will be
operator-irreducible.

Note that for an operator-irreducible integral $\mathcal{I} =
L(x,y)\varphi (x)$ in the process of the work of the algorithm one
can find points $y_1$, $y_2$, \ldots , $y_s$, $s \leq n+1$, and
operators $P_1(x)$, \ldots , $P_s(x)$ such that
\begin{equation}\label{P0}
    P_1(x)L(x,y_1) + \ldots + P_s(x)L(x,y_s) = 1.
\end{equation}
In fact one can take generic $y_i$, i.e.\ some points which don't
coincide with the poles of the coefficients of all operators
appearing in the Euclidean algorithm and don't make their leading
coefficients identically zero. The operators $P_i(x)$ are found
using the standard reverse substitution in the equations (\ref{E1}).

One should note that in the general case the number of the points
$y_i$ can not be made smaller than $n+1$. Below we construct an
example of such second-order integral $\mathcal{I} = L(x,y)\varphi
(x)$ that $\forall y_1, y_2$, $rGCD(L(x,y_1),
L(x,y_2))=B_k(x,y_1,y_2)$, $\ord B_k(x,y_1,y_2)=1$, but there exist
points $ y_1, y_2,y_3$, such that $rGCD(L(x,y_1),
L(x,y_2),L(x,y_3))=B^{(1)}_r(x,y_1,y_2,y_3) = 1$.

To this end let us consider the following three-dimensional linear
space of second degree polynomials:
$$
V = \{ p(x)=\alpha \cdot 1 + \beta \cdot x + \gamma \cdot x^2 | \
\alpha,\, \beta,\, \gamma \in \mathbb{R} \}.
$$
Introduce the following two-dimensional subspaces $W_y$ depending on
a parameter $y$ which are spanned by the basis
$$
\begin{array}{l}
  z_1 = y \cdot 1 + (y+1)\cdot x, \\
  z_2 = y \cdot x + (y+1)\cdot x^2,
\end{array}
 \qquad W_y = \langle z_1, z_2 \rangle.
$$
As one can readily check, for any fixed $y$ the functions $z_1$,
$z_2$ are linearly independent and for every two distinct values of
the parameter $y_1 \neq y_2$, $W_{y_1}$ and $W_{y_2}$ have a
one-dimensional intersection. On the other hand $W_{y_1} \cap
W_{y_2} \cap W_{y_3} = \{ 0 \}$. This allows us to construct the
following LODO
$$
L(x,y)\varphi(x)= \left|
\begin{array}{ccc}
  \varphi & \varphi' & \varphi'' \\
  z_1 & z'_1 & z''_1 \\
  z_2 & z'_2 & z''_2
\end{array}
 \right|
\cdot \left|
\begin{array}{cc}
  z_1 & z'_1  \\
  z_2 & z'_2
\end{array}
 \right| ^{-1}  =
$$
$$
= \left[ D^2 - \frac{2(y+1)}{xy+x+y} D
+\frac{2(y+1)^2}{(xy+x+y)^2}\right]\varphi(x),
$$
with $rGCD(L(x,y_1), L(x,y_2))=B_1(x,y_1,y_2)=
 (y_1-y_2)\big( (xy_1+x+y_1)(xy_2+x+y_2) D
  - (2 x y_1 y_2 + 2 x y_1 + 2 x y_2 + 2 x + 2 y_1 y_2 + y_1 + y_2 \big)
$, but $rGCD(L(x,y_1), L(x,y_2),L(x,y_3))= 1$. So one concludes that
in our algorithm $\mathbf{RI}$ for the constructed operator $L(x,y)$
one can not find two points $y_1$, $y_2$ and operators $P_1(x)$,
$P_2(x)$, such that $P_1(x)L(x,y_1) + P_2(x)L(x,y_2) = 1$ although
the integral $\mathcal{I} = L(x,y)\varphi (x)$ is
operator-irreducible. On the other hand for three generic points
$y_1$, $y_2$, $y_3$ the equality (\ref{P0}) holds. One can easily
construct in the same way an example on an integral of order $n$ for
which one has to choose at least $n+1$ points $y_i$ in order to get
(\ref{P0}).

\begin{theorem}\label{teorema-1}
The three definitions~\ref{def-red-ord}, \ref{def-red-oper} and
\ref{def-red-ker} of reducibility of Euler integrals are equivalent.
\end{theorem}
\begin{proof}
If an integral $\mathcal{I}$ is operator-reducible, i.e.\ if
$\mathcal{I}= L(x,y)\varphi (x) =M(x,y) R(x)\varphi (x)$ with some
non-trivial $R(x)$ then the integral $\mathcal{I}$ has a nontrivial
kernel which at least contains the space of solutions of the
equation $R(x)\varphi (x)=0$. Hence $\mathcal{I}$ is
kernel-reducible.

If an integral $\mathcal{I}$ is kernel-reducible, that is there
exists a non-zero function $\varphi (x)$ such that $\mathcal{I} =
L(x,y)\varphi (x) \equiv 0$ then obviously $L(x,y)$ is representable
in the form
$$
L(x,y) =L_1(x,y)\left(\frac{d}{dx} - \frac{\varphi _x' }{\varphi
}\right),
$$
so the integral $\mathcal{I}$ is operator-reducible.

It is obvious that operator-reducible integrals are order-reducible.
In order to prove the converse we need the following Lemma:
\begin{lemma}\label{lemma-1}
Let two Euler integrals $\mathcal{I} = L(x,y)\varphi (x)$ and
$\mathcal{I}_1 = M(x,y)\psi (x)$ of the form (\ref{2.1}) are given,
and $\mathcal{I}_1$ is operator-irreducible. Let the following
inclusion for the function stocks generated by the integrals hold:
$Z(\mathcal{I})\subseteq Z(\mathcal{I}_1)$. Then the order of
$\mathcal{I}$ is greater or equal to the order of $\mathcal{I}_1$,
the integral $\mathcal{I}$ can be operator-reduced to the integral
$\mathcal{I}_1$  (i.e. $L(x,y) = M(x,y) R(x) $ for some operator
$R(x)$), and the function stocks generated by the integrals
$\mathcal{I}$ and $\mathcal{I}_1$ coincide.
\end{lemma}
\begin{proof}
Applying the algorithm $\mathbf{RI}$ to the integral $\mathcal{I}_1
= M(x,y)\psi (x)$ we find points $y_1$, $y_2$, \ldots , $y_s$ and
operators $P_1(x)$, \ldots , $P_s(x)$, such that (\ref{P0}) holds
for $M(x,y)$, since $\mathcal{I}_1$ is already operator-irreducible.
Substituting the points $y_i$, $i = 1,\ldots , s$, into the operator
$L(x,y)$ which generates the integral $\mathcal{I}$, we obtain $s$
LODO $L(x,y_i)$ with coefficients depending on $x$.

For any function $ \varphi (x)$ there exists a function $ \psi (x)$,
such that
$$ L(x,y)\varphi (x) = M(x,y)\psi (x),
$$
since $Z(\mathcal{I})\subseteq Z(\mathcal{I}_1)$. Hence
\begin{equation}\label{2.4}
    L(x,y_i)\varphi (x) = M(x,y_i)\psi (x)
\end{equation}
for $i=1, \ldots , s$. Multiplying the equalities (\ref{2.4}) with
the corresponding $P_i(x)$ and adding we obtain
$$ \sum_{i=1}^{s}P_i(x) L(x,y_i)\varphi (x) =\sum_{i=1}^{s}P_i(x) M(x,y_i)\psi
(x),
$$
or in view of (\ref{P0}) for the operator $M(x,y)$, we obtain
$$ \sum_{i=1}^{s}P_i(x) L(x,y_i)\varphi (x) =\psi(x).
$$
This means that  $\psi(x)$ is differentially expressible in terms of
$\varphi (x)$: $\psi(x) = R(x)\varphi (x)$, $R(x)=
\sum_{i}P_i(x)L(x,y_i)$, so $L(x,y)\varphi (x) = M(x,y)R(x)\varphi
(x)$, i.e.\ $\mathcal{I}$ is operator-reducible to $\mathcal{I}_1$.
Since for every function $\psi(x)$ there exists a solution $\varphi
(x)$ of an ordinary differential equation $\psi(x) = R(x)\varphi
(x)$, we have $M(x,y)\psi (x)= M(x,y)R(x)\varphi (x)=L(x,y)\varphi
(x)$ for any $\psi(x)$ and the corresponding $\varphi (x)$.
Consequently $Z(\mathcal{I}) = Z(\mathcal{I}_1)$.
\end{proof}

\emph{The end of the proof of Theorem~\ref{teorema-1}.}

Let an integral $\mathcal{I} = L(x,y)\varphi (x)$ be order-reducible
to some $\mathcal{I}_1 = L_1(x,y)\psi (x)$. We must show that
$\mathcal{I}$ is operator-reducible. One can suppose $\mathcal{I}_1$
to be already operator-irreducible since in the opposite case one
can use the algorithm $\mathbf{RI}$ to find an operator-irreducible
$\mathcal{I}_2 = M(x,y)\xi (x)$ with the same function stock. Since
the function stocks generated by $\mathcal{I}_1$ and $\mathcal{I}_2$
coincide, we can apply Lemma~\ref{lemma-1} to the integrals
$\mathcal{I}$ and $\mathcal{I}_2$ and obtain that $L(x,y) =
M(x,y)R(x)$, i.e.\ the integral $\mathcal{I}$ is operator-reducible.
\end{proof}

In view of the obtained result we will call integrals simply
``reducible'' or ``irreducible'' without specifying the precise
meanings given in the definitions~\ref{def-red-ord},
\ref{def-red-oper} and \ref{def-red-ker}.
\begin{corollary}\label{sl-p6}
The kernel of a Euler integral $\mathcal{I} = L(x,y)\varphi (x)$
coincides with the kernel of the operator $R(x)$, obtained in the
algorithm~\textbf{RI}.
\end{corollary}
\begin{proof} From (\ref{2.3}) one can easily see that $Ker \, R(x) \subseteq Ker\,
\mathcal{I}$. Let $\varphi (x)\in Ker\, \mathcal{I}$. Then
$0=L(x,y)\varphi (x)=M(x,y)R(x)\varphi (x)=M(x,y)\psi(x)$,
$\psi=R\varphi$. From this we conclude that $\psi(x)=0$, since
$M(x,y)$ is irreducible. Consequently $\varphi (x)\in Ker\, R(x)$.
\end{proof}

We are summing up the results about the structure of reducible
integrals obtained so far in the following Proposition:
\begin{proposition}\label{Utv-9a}
Among Euler integrals generating the same function stocks there is a
unique  (up to gauge transformations $L_0\mapsto L_0(x,y)g(x)$)
irreducible integral $\mathcal{I}_0 = L_0(x,y)\psi (x)$. All other
integrals with the same function stock have the following form:
$\mathcal{I} = L_0(x,y)R(x)\varphi (x)$, where $R(x)$ is an
arbitrary LODO with coefficients depending on $x$ only.
\end{proposition}
\begin{proof}
Let $Z(\mathcal{I}_1) = Z(\mathcal{I}_2)$. Using the
algorithm~$\mathbf{RI}$ if necessary we may assume that
$\mathcal{I}_1$ and $\mathcal{I}_2$ are irreducible. Applying
Lemma~\ref{lemma-1} we see that $L_1(x,y) = L_2(x,y)R_1(x)$,
$L_2(x,y) = L_1(x,y)R_2(x)$, so $R_1(x)=g(x)$ must be invertible in
the ring of LODO so it must be a zero-order operator. If
$\mathcal{I}$ is an arbitrary (reducible) integral and
$\mathcal{I}_0$ is irreducible, we can use Lemma~\ref{lemma-1} again
obtaining the required identity $L(x,y) = L_0(x,y)R(x)$.
\end{proof}

\begin{theorem}\label{Teorema.2.10}
Let $\mathcal{I}$ and $\mathcal{I}_1$ be Euler integrals of the
form~(\ref{2.1}). Then the function stocks generated by these
integrals either coincide or have a finite-dimensional intersection.
\end{theorem}
\begin{proof}
We again can assume that the given integrals $\mathcal{I}$ and
$\mathcal{I}_1$ are irreducible. Let $\mathcal{I} = L(x,y)\varphi
(x)$, $\mathcal{I}_1 = L_1(x,y)\psi (x)$ and some function $f(x,y)$
is contained in the intersection of the corresponding function
stocks. Then the following equality holds:
\begin{equation}\label{2.5}
f(x,y)= L(x,y)\varphi (x) = L_1(x,y)\psi (x)
\end{equation}
for some $\varphi (x)$, $\psi (x)$. Apply the algorithm
$\mathbf{RI}$ to the integral $\mathcal{I}$. In the process we
obtain a finite number of points $y_1$, $y_2$, \ldots , $y_s$ and
operators $P_1(x)$, \ldots , $P_s(x)$, such that~(\ref{P0}) holds.
Substituting the points $y_i$ into (\ref{2.5}), making the
appropriate linear combination and using (\ref{P0}) we obtain
\begin{equation}\label{2.7}
\varphi (x) = N(x)\psi (x),
\end{equation}
where $N(x)= \sum_{i=1}^{s}P_i(x) L_1(x,y_i)$. Substituting
(\ref{2.7}) into (\ref{2.5}) we get
\begin{equation}\label{2.8}
L(x,y)N(x)\psi (x) = L_1(x,y)\psi (x)
\end{equation}
or
\begin{equation}\label{2.9}
\Big(L(x,y)N(x) - L_1(x,y)\Big)\psi (x) = 0.
\end{equation}
If the operator $L(x,y)N(x) - L_1(x,y)$ vanishes identically then
any function $\psi (x)$ satisfies (\ref{2.9}), and, consequently,
also (\ref{2.8}). From this we conclude that the function stocks
generated by $\mathcal{I}$ and $\mathcal{I}_1$, coincide
(Lemma~\ref{lemma-1}). If the operator $L(x,y)N(x) - L_1(x,y)$ does
not vanish, we deduce from (\ref{2.9}) that $\psi (x)$ belongs to
the kernel of the integral, generated by the operator
$F(x,y)=L(x,y)N(x) - L_1(x,y)$. Applying the algorithm $\mathbf{RI}$
to $F(x,y)$ one obtains the decomposition $F(x,y)=F_1(x,y)S(x)$,
where $F_1(x,y)$ is irreducible and $\psi (x)$ belongs to the
finite-dimensional space of solutions of the ordinary differential
equation $S(x)\psi (x)=0$. From (\ref{2.5}) we conclude that the
intersection of the function stocks, generated by the integrals
$\mathcal{I}$ and $\mathcal{I}_1$ is the image of this
finite-dimensional space under the action of the operator
$L_1(x,y)$. This concludes the proof.
\end{proof}

\textbf{Remark.} Obviously the proofs of the
Proposition~\ref{Utv-9a} and Theorem~\ref{Teorema.2.10} imply that
we have a simple algorithmic way to check if the function stocks
generated by two given integrals $\mathcal{I}_1$ and $\mathcal{I}_2$
coincide. For this we find the corresponding irreducible operators
$L_1(x,y)$ and $L_2(x,y)$ with the same function stocks using the
algorithm~$\mathbf{RI}$. If the orders of $L_1(x,y)$, $L_2(x,y)$ do
not coincide then the function stocks $\mathcal{I}_1$ and
$\mathcal{I}_2$ are different. If the orders of $L_1(x,y)$ and
$L_2(x,y)$ coincide, we divide $L_1(x,y)$ by $L_2(x,y)$ on the left:
$$
 L_1(x,y) = L_2(x,y)g (x,y) + B(x,y), \qquad \ord B < \ord L_1=\ord L_2.
$$
The function stocks of $\mathcal{I}_1$ and $\mathcal{I}_2$ coincide
if and only if $B(x,y) \equiv 0$, $g _y \equiv 0$. In order to find
the dimension of the intersection $Z(\mathcal{I}_1) \cap
Z(\mathcal{I}_2)$ we repeat the procedure described in the proof of
Theorem~\ref{Teorema.2.10}, i.e.\ we form the operators in the left
hand side of~(\ref{2.9}) and apply the algorithm~\textbf{RI} to
split off the maximal right divisor $S(x)$ of the operator $F(x,y)$.
The dimension of the intersection $Z(\mathcal{I}_1) \cap
Z(\mathcal{I}_2)$ will be equal to the order of the operator $S(x)$.

\section{Multiintegrals}

In this Section we prove the equivalence of two possible definitions
of reducibility of Euler multiintegrals (\ref{3.1}) and give the
algorithms of reduction and other necessary technical stuff.

\begin{definition}\label{def-zapas-multi}
    The function stock $Z(\mathcal{J})$, generated by a multiintegral
(\ref{3.1}) is the set of all functions of two variables $(x,y)$
which is obtained after the substitution of arbitrary (smooth)
functions $\varphi_1 (x)$, \ldots , $\varphi_k(x)$ into (\ref{3.1}).
\end{definition}
\begin{definition}\label{def-red-oper-multi}
    We will call a multiintegral (\ref{3.1}) \emph{operator-reducible},
if one can represent its operator row $\big(L_1(x,y), \ldots ,
L_k(x,y)\big)$ as a composition of operator matrices
\begin{equation}\label{3.3}
 \big(L_1(x,y), \ldots , L_k(x,y)\big) = \big(M_1(x,y), \ldots , M_p(x,y)\big)
  \left(
 \begin{array}{ccc}
   R_{11}(x) & \ldots & R_{1k}(x) \\
    \vdots & \ddots & \vdots \\
   R_{p1}(x) & \ldots & R_{pk}(x)
    \end{array}
  \right),
\end{equation}
where $R_{ij}(x)$ are differential operators with coefficients
depending on $x$ only, $M_i(x,y)$ are differential operators with
coefficients depending on $(x,y)$; the matrix $\big(R_{ij}(x)\big)$
is required to be non-invertible and $p \leq k$.
\end{definition}
Below we will briefly represent the formula (\ref{3.3}) as
\begin{equation}\label{3.4}
     \overline{L}(x,y) = \overline{M}(x,y)\widehat R(x).
\end{equation}
Here and below $\widehat R(x)$ stands for a matrix of differential
operators (not necessarily a square matrix). Invertibility of such a
matrix is understood in the algebraic sense, i.e.\ as existence of
another operator matrix $\widehat S(x)$, such that $\widehat R
\widehat S = \widehat S \widehat R =\big(\delta_{ij}\big)$.
\begin{definition}\label{def-red-jadro-multi}
We will call a multiintegral (\ref{3.1}) \emph{kernel-irreducible}
(or we say that the multiintegral has no kernel) if $\mathcal{J} =
L_1(x,y)\varphi_1(x) + \ldots + L_k(x,y)\varphi_k(x)$ vanishes as a
function of two variables only for $\varphi_i(x)\equiv 0$, $i=1,
\ldots , k$. In the contrary case the multiintegral $\mathcal{J}$
will be called \emph{kernel-reducible}.
\end{definition}
Note that we do not find appropriate to introduce here an analogue
of \emph{order-reducibility} of multiintegrals. As we will see
below, even irreducible multiintegrals with the same function stocks
admit a representation (\ref{3.1}) with operators $L_i(x,y)$ of
arbitrary high order. This follows from the fact that for arbitrary
multiintegral one can find representations (\ref{3.4}) with
\emph{invertible} matrices $\widehat R(x)$ of operators $R_{ij}(x)$
of arbitrary high order.

One can perform the following operations on the matrix $\widehat
R(x)$: transposition of two rows and addition to one of the row of
another  row multiplied on the left with an ordinary differential
operator with coefficients depending only on $x$.

The following algorithm uses these operations for reduction of the
matrix $\widehat R(x)$ to a special convenient form.

\begin{center}
    \textbf{Algorithm} $\mathbf{RM}$
\end{center}

First of all, note that the matrix $\widehat R(x)$ in the
representation~(\ref{3.4}) may be considered to be free from zero
rows and zero columns.

If the first column has only one non-zero element, transpose the
rows and put it into the first (uppermost) place, then start
processing of the second column. If there are several non-zero
elements in the first column, we perform the division of one of them
with another one with the remainder, for example
$$R_{11}(x) = T(x)R_{21}(x) + Q(x).
$$
Then, if one subtracts the second row of the matrix $\widehat R(x)$
multiplied on the left with the differential operator $T(x)$ from
the first row, one obtains an equivalent matrix, which has the
operator $Q(x)$ of lower order than the previous entry $R_{11}(x)$.
It is obvious that such operations correspond to the operations of
the Euclidean algorithm (\ref{E1}) and result in an equivalent
matrix with $R^{(1)}_{11}(x)= rGCD(R_{11},R_{21})$, and
$R^{(1)}_{21}\equiv 0$. Consecutively applying this procedure to the
other non-zero elements of the first column we obtain a matrix which
has only one non-zero element $R^{(2)}_{11}(x)= rGCD(R_{11}, \ldots
, R_{p1})$ in the first column.

Next we proceed to the second column. If it has some of the entries
$R_{22}$, $R_{32}$, \ldots , $R_{p2}$ different from zero, we apply
to them the same procedure as above and obtain $R^{(2)}_{22}(x)=
rGCD(R_{22}, \ldots , R_{p2})$. If, contrarily, all of $R_{22}$,
\ldots , $R_{p2}$ vanish, we start processing of the elements
$R_{23}$, $R_{33}$, \ldots , $R_{p3}$ of the third column. Obviously
this procedure is a direct non-commutative analogue of the standard
Gauss elimination algorithm, so in a finite number of steps we
obtain a matrix of the form
$$
\widehat R_{\texttt{st}}(x) = \left(
\begin{array}{cccccccc}
  R_{11}(x) & \ldots & * & * & \ldots & * & * & \ldots \\
  0 & \ldots & 0& R_{2q_2}(x)  & \ldots & * & * &   \\
  0 & \ldots & 0 & 0 & \ldots & 0 &  R_{3q_3}(x) & \ldots \\
  \vdots & \vdots &\vdots & 0 & \vdots & 0 &  0 & \ddots \\
  0 & \ldots & \ldots &  \ldots &  \ldots & \ldots  &  \ldots &
\end{array}
\right),
$$
which will be called below \emph{echelon matrix}. In a more general
case (which is of no interest to us) echelon matrices may have a few
zero starting columns. First non-zero elements of each row
$R_{11}(x)$, $R_{2q_2}(x)$, $R_{3q_3}(x)$, \ldots \  will be called
(as usual) \emph{pivots}. The number of non-zero rows of the
obtained echelon matrix is called the \emph{rank} of the initial
operator matrix. As proved in \cite{Jacobson,Lopat}, the rank is an
invariant of a matrix and does not depend on the method of its
reduction to an echelon form.

The procedure of reduction to the echelon form is equivalent to
multiplication of the initial matrix $\widehat R(x)$ on the left
with some invertible matrix $\widehat T(x)$. Therefore in the
formula (\ref{3.4}) we can substitute $\overline{M}(x,y)$ with
$\overline{M}(x,y)\widehat T^{-1}(x)$ and $\widehat R(x)$ with
$\widehat R_{\texttt{st}}(x)=\widehat T(x)\widehat R(x)$; so we can
assume below that the matrix $\widehat R(x)$ already is echelon
matrix.

It is easy to see that the matrix $\widehat R(x)$ is invertible if
and only if it is a square matrix and after its reduction to echelon
form we obtain an upper-triangular matrix with non-zero diagonal
elements which are invertible elements in the ring of differential
operators, that is if they are zero-order operators (functions). The
elements strictly above the main diagonal may be arbitrary
differential operators.

Analogously to the algorithm $\mathbf{RI}$ we give below another
algorithm which extracts the ``maximal'' right factor $\widehat
R(x)$ from a given multiintegral, i.e.\ for a given operator row
$\overline{L}(x,y)$ a representation (\ref{3.4}) will be found, with
a row $\overline{M}(x,y)$ defining an operator-irreducible
multiintegral with the same function stock.

\begin{center}
    \textbf{Algorithm} $\mathbf{RMI}$
\end{center}

The initial data of the algorithm is an operator row $\big(L_1(x,y),
\ldots , L_k(x,y)\big)$.

According to the decomposition (\ref{3.3}) with an echelon matrix
$\widehat R(x)$ we see that one has to find operators $M_i(x,y)$,
$R_{ij}(x)$, such that
\begin{eqnarray}
 L_1(x,y)&= &M_1(x,y)R_{11}(x),\label{3.5.1} \\
 L_2(x,y)&= &M_1(x,y)R_{12}(x) + M_2(x,y)R_{22}(x),\label{3.5.2} \\
 \ldots & & \ldots \nonumber
\end{eqnarray}
Using the algorithm $\mathbf{RI}$ we find the decomposition
(\ref{3.5.1}), where $R_{11}(x)$ is a LODO of maximal order and
$M_1(x,y)$ is irreducible. Also as a byproduct of the algorithm
$\mathbf{RI}$ we get points $y_1$, \ldots , $y_s$ and operators
$P_1(x)$, \ldots , $P_s(x)$, such that
\begin{equation}\label{3.5.3}
  P_1(x) M_1(x,y_1) + \ldots + P_s(x) M_1(x,y_s) = 1.
\end{equation}
Now we should find  $R_{12}(x)$, $M_2(x,y)$, $R_{22}(x)$ such that
the equality (\ref{3.5.2}) holds. We obtain this decomposition in
three stages.

\ \ \ \ \emph{Stage A. Finding $R_{22}(x)$.}

Substituting the found points $y_i$ into (\ref{3.5.2}) we obtain
$$
L_2(x,y_i)= M_1(x,y_i)R_{12}(x) + M_2(x,y_i)R_{22}(x).
$$
Multiplying these identities with the operators $P_i(x)$ on the left
and adding, we get
\begin{equation}\label{3.5.5}
    \left[\sum_{i=1}^{s}P_i(x)M_1(x,y_i)\right]R_{12}(x) +
    \left[\sum_{i=1}^{s}P_i(x)M_2(x,y_i)\right]R_{22}(x)=
    \sum_{i=1}^{s}P_i(x)L_2(x,y_i),
\end{equation}
or, in view of (\ref{3.5.3}),
\begin{equation}\label{3.6.1}
  R_{12}(x) +
    \left[\sum_{i=1}^{s}P_i(x)M_2(x,y_i)\right]R_{22}(x)=
    \sum_{i=1}^{s}P_i(x)L_2(x,y_i).
\end{equation}
Multiplying (\ref{3.6.1}) on the left with $M_1(x,y)$ and
subtracting from (\ref{3.5.2}), we arrive at:
\begin{equation}\label{3.6.2}
     \begin{array}{l}
 \left[M_2(x,y) - M_1(x,y)
 \sum_{i}P_i(x)M_2(x,y_i)\right]R_{22}(x)\\[1em]
 \ \ \ \ \ \ \ \ \ \ \ \ {}=
 L_2(x,y) -M_1(x,y)\sum_{i}P_i(x)L_2(x,y_i).
     \end{array}
\end{equation}
The right hand side of this equality is some known operator
$N(x,y)$. Hence in order to find $R_{22}(x)$ one has to apply the
algorithm $\mathbf{RI}$ to $N(x,y)$, obtaining the decomposition
\begin{equation}\label{3.6.2'}
    K(x,y)R_{22}(x) =N(x,y)
\end{equation}
with $R_{22}(x)$ of maximal order; the obtained operator $K(x,y)$ is
irreducible.

\ \ \ \ \emph{Stage B. Finding $R_{12}(x)$.}

Rewrite (\ref{3.6.1}) as
\begin{equation}\label{3.6.1'}
 R_{12}(x) + S(x)   R_{22}(x)= T(x)\equiv \sum P_i(x)L_2(x,y_i),
\end{equation}
where the operators $T(x)$ and $R_{22}(x)$ are known. The operator
$S(x)$ can be arbitrarily chosen (this corresponds to multiplication
of the echelon matrix $\widehat R(x)$ on the left with an invertible
upper-triangular matrix) and define the corresponding $R_{12}(x)$.

\ \ \ \ \emph{Stage C. Finding $M_{2}(x,y)$.}

Recall that  (\ref{3.6.2}) was obtained from (\ref{3.6.1}) and
(\ref{3.5.2}). Now we do the inverse operation: multiply
(\ref{3.6.1'}) with $M_1(x,y)$ on the left and add (\ref{3.6.2'}) to
it:
$$M_1(x,y)R_{12}(x) + \left[M_1(x,y)S(x) + K(x,y)\right]R_{22}(x)=
 L_2(x,y).
$$
Thus we obtain the required equality (\ref{3.5.2}) with $M_2(x,y)
 =M_1(x,y)S(x) + K(x,y)$.

Note that since the obtained $K(x,y)$ is irreducible we can find
points $y_{2,1}$, $y_{2,2}$, \ldots , $y_{2,s_2}$ and operators
$P_{2,1}(x)$, \ldots , $P_{2,s_2}(x)$, such that
\begin{equation}\label{3.7}
 \sum_{i=1}^{s_2}P_{2,i}(x)K(x,y_{2,i}) = 1.
\end{equation}

If the number of operators $L_i(x,y)$ in the row $\overline{L}(x,y)$
is greater than two, we can analogously find the required operators
$M_j(x,y)$, $R_{ij}(x)$, $j \geq 3$. We demonstrate this below for
$j=3$.

In the equality
\begin{equation}\label{3.8}
    L_3(x,y)=M_1(x,y)R_{13}(x) + M_2(x,y)R_{23}(x) + M_3(x,y)R_{33}(x)
\end{equation}
the operators $M_1(x,y)$ and $M_2(x,y)$ are known. Using the points
$y_1$, \ldots , $y_s$ chosen above and (\ref{3.5.3}) we arrive at
the following analogues of (\ref{3.6.1}) and (\ref{3.6.2}):
\begin{equation}\label{3.11.15}
K(x,y)R_{23}(x) + \left[M_3(x,y) -M_1(x,y)
\sum_{i}P_i(x)M_3(x,y_i)\right]R_{33}(x)={}
\end{equation}
$$
{}= L_3(x,y) -M_1(x,y)\sum_{i}P_i(x)L_3(x,y_i).
$$
Using now the points $y_{21}$, \ldots , $y_{2s_2}$ and the equality
(\ref{3.7}), we can cancel the term $K(x,y)R_{23}(x)$ and obtain an
analogue of (\ref{3.6.2'}):
$$
K_3(x,y)R_{33}(x) = N_3(x,y)
$$
with the known operator $ N_3(x,y)$. This allows us to find
$R_{33}(x)$ using the algorithm $\mathbf{RI}$. Repeating \emph{stage
B} we find $R_{13}(x)$, $R_{23}(x)$ with some freedom, analogous to
the freedom of choice of $S(x) $ in (\ref{3.6.1'}). Repeating the
\emph{stage C} one finds $ M_3(x,y)$.

\textbf{Remark}. It may happen that the right hand side of
(\ref{3.6.2}) vanishes. Then we set $R_{22}(x)\equiv 0$ and find
$R_{12}(x)$ from (\ref{3.6.1}). Thus $ M_2(x,y)$ is actually absent
in (\ref{3.5.2}) and we can set $R_{33}(x)\equiv 0$ on the following
step in (\ref{3.8}); in this case the pivot element becomes
$R_{23}(x)$. Thus $R_{ij}(x)$ becomes a non-square echelon matrix.
The number of operators $ M_i(x,y)$ in this case will be smaller
that the number of the initial operators $ L_i(x,y)$, i.e. $p<k$.

As the result of the algorithm $\mathbf{RMI}$ we obtain an operator
row $\overline{M}(x,y) = \big(M_1(x,y), \ldots , M_p(x,y)\big)$ and
an operator  echelon matrix $\widehat{R}(x)$ in the representation
(\ref{3.3}). It is not obvious that the complete multiintegral
$\mathcal{J}_1 = \overline{M}(x,y)\mathring{\psi}(x)$ is
operator-irreducible. This fact will be proved below
(Corollary~\ref{3.6-}). Inclusion ${Z}\big(\mathcal{J}\big)
\subseteq {Z}\big(\mathcal{J}_1\big)$ is obvious. The inverse
inclusion follows from the simple fact that the echelon matrix
$\widehat{R}(x)$ obtained in the algorithm $\mathbf{RMI}$ has no
zero rows so the system of linear ordinary differential equations
$\widehat{R}(x)\mathring{\varphi} = \mathring{\psi}$ for the unknown
$\mathring{\varphi}$ is solvable for any right hand side
$\mathring{\psi}$.

\textbf{Remark}. If the given multiintegral was operator-irreducible
then the algorithm \textbf{RMI} will given the decomposition
(\ref{3.3}) with $k=p$ and a triangular invertible matrix
$\widehat{R}(x)$ as the output. The ambiguity of the choice of the
operator $S(x)$ on \emph{stage B} will correspond to the possibility
to write a representation (\ref{3.3}) with arbitrary triangular
invertible matrix $\widehat{R}(x)$ for any multiintegral
$\mathcal{J}$. Note that the requirement of non-invertibility of the
matrix $\widehat{R}(x)$ without the requirement $p \leq k$ in the
definition~\ref{def-red-oper-multi} is not sufficient: it is easy to
give an example of decomposition (\ref{3.3}) with $p>k$ (so with a
non-invertible matrix $\widehat{R}(x)$) for any multiintegral
$\mathcal{J}=\overline{L}(x,y)\mathring{\varphi}(x)$.

\begin{proposition}\label{utver-nered-M}
The multiintegral $\big(M_1(x,y), \ldots , M_p(x,y)\big)$ obtained
in the process of the work of the algorithm $\mathbf{RMI}$ is
kernel-irreducible.
\end{proposition}

\begin{proof}
Let $ \big( \psi_1(x),\ldots ,\psi_p(x)\big)$ be a nonzero element
of the kernel of the multiintegral $\overline{M} = \big(M_1(x,y),
\ldots , M_p(x,y)\big)$:
\begin{equation}\label{3.10.1}
\overline{M}\mathring{\psi}\equiv 0.
\end{equation}

By construction of the operators $M_i(x,y)$ they satisfy
\begin{equation}\label{3.10.2}
\overline{L} = \overline{M} \widehat{R}.
\end{equation}
First we carry out the proof of the Proposition for the number of
elements in the operator rows $\overline{L}$ and $\overline{M}$
equal to 2. In this case (\ref{3.10.1}) means that $M_i(x,y)$,
$\psi_i(x)$ , $i=1,2$, satisfy the equation
\begin{equation}\label{3.10.1'}
0=M_1(x,y)\psi_1 + M_2(x,y)\psi_2.
\end{equation}
The operators $M_1(x,y)$, $M_2(x,y)$ satisfy in turn
\begin{equation}\label{3.10.2'}
L_2(x,y)=M_1(x,y)R_{12}(x) + M_2(x,y)R_{22}(x),
\end{equation}
where $R_{12}(x)$, $R_{22}(x)$ are found in the algorithm
\textbf{RMI}. Now we carry out with the equality (\ref{3.10.1'}) the
same operations as we did in the algorithm \textbf{RMI} for the
equation (\ref{3.10.2'}) in order to find $R_{12}(x)$, $R_{22}(x)$,
i.e.\ we take the same points $y_i$ and the same operators $P_i(x)$
which were found in the algorithm \textbf{RMI}, then we compose the
same linear combination of the equations that in (\ref{3.6.1}). We
get
\begin{equation}\label{3.11.1}
    \psi_1(x) + \left[\sum_{i}P_i(x)M_2(x,y_i)\right]\psi_2(x) = 0.
\end{equation}
From (\ref{3.11.1}) and (\ref{3.10.1'}) we have
\begin{equation}\label{3.11.2}
\left[M_2(x,y)- M_1(x,y)\sum_{i}P_i(x)M_2(x,y_i)\right]\psi_2(x) = 0
\end{equation}
(an analogue of (\ref{3.6.2}) in the algorithm \textbf{RMI}). Thus
$\psi_2(x)\equiv 0$, since the algorithm \textbf{RI} gave an
irreducible integral $K(x,y)$ in the left hand side of this
equality.

Then from (\ref{3.10.1'}) we get $M_1(x,y)\psi_1\equiv 0$. But
$M_1(x,y)$ is irreducible by construction so $\psi_1(x)\equiv 0$.

Thereby the Proposition is proved for the multiintegral $\mathcal{J}
= \big(M_1(x,y),M_2(x,y)\big)$.

In order to prove the Proposition for the general case
$\mathcal{J}=\big(M_1(x,y), \ldots , M_p(x,y)\big)$ we analogously
consider the equality
\begin{equation}\label{3.12.2}
0 = M_1(x,y)\psi_1(x) + \ldots + M_p(x,y)\psi_p(x),
\end{equation}
where $(\psi_1(x), \ldots ,\psi_p(x))$ is an element of the kernel
of the multiintegral $\mathcal{J}$ instead of the equality
\begin{equation}\label{3.12.1}
L_k(x,y) = M_1(x,y)R_{1k}(x) + \ldots + M_p(x,y)R_{pk}(x).
\end{equation}
Applying to (\ref{3.12.2}) all operations described in the algorithm
\textbf{RMI} one arrives at $K_p(x,y)\psi_p(x)=0$, where $K_p(x,y)$
is irreducible, which implies $\psi_p(x)=0$. Performing the reverse
run of the algorithm we find consecutively $\psi_{p-1}(x)=0$, \ldots
, $\psi_{1}(x)=0$.
\end{proof}
\begin{corollary}\label{sleds.3.13}
The kernel of the multiintegral $\mathcal{J}=\big(L_1(x,y), \ldots ,
L_k(x,y)\big)$ coincides with the kernel of the matrix
$\widehat{R}(x)$ obtained in the algorithm \textbf{RMI}.
\end{corollary}
\begin{proof}
The algorithm gives us $\overline{L}(x,y) \mathring{\varphi}(x)=
\overline{M}(x,y)\widehat R(x)\mathring{\varphi}(x)$, where
$\overline{M}(x,y)$ defines a kernel-irreducible multiintegral. If
$\mathring{\varphi}(x)$ belongs to the kernel of the matrix
$\widehat R(x)$ then $\mathring{\varphi}(x)$ obviously belongs to
the kernel of the multiintegral $\mathcal{J}$. Conversely if
$\mathring{\varphi}(x)$ belongs to the kernel of the multiintegral,
i.e.\ $\overline{L}(x,y) \mathring{\varphi}(x)=0$, then
$\overline{M}(x,y)\mathring{\psi}(x)=0$, $\mathring{\psi}=\widehat
R(x)\mathring{\varphi}(x)$ and in virtue of the
kernel-irreducibility of $\overline{M}(x,y)$,
$\mathring{\psi}=\widehat R(x)\mathring{\varphi}(x)=0$, so
$\mathring{\varphi}(x)$ belongs to the kernel of the matrix
$\widehat R(x)$.
\end{proof}
\begin{proposition}\label{Utv.3.14}
A multiintegral $\mathcal{J}$ has an infinite-dimensional kernel if
and only if the matrix $\widehat R(x)$ obtained in the algorithm
\textbf{RMI} is non-square (the number of its columns is greater
than the number of its rows).
\end{proposition}
\begin{proof}
Suppose that the matrix $\widehat R(x)$ in the decomposition
(\ref{3.4}) obtained by the algorithm \textbf{RMI} is a square
matrix. By Corollary~\ref{sleds.3.13} the kernel of the
multiintegral coincides with the kernel of $\widehat{R}$. The
dimension of the space of columns $(\varphi_1(x)$, \ldots ,
$\varphi_k(x))$  which are solutions of the matrix differential
equation $\widehat R(x)\mathring{\varphi}(x)=0$ is finite. In fact,
since $\widehat R(x)$ is a square echelon matrix without zero rows,
we see that $\varphi_k(x)$ satisfies the equation
$R_{kk}(x){\varphi}_k(x)=0$ and consequently belongs to a
finite-dimensional space. Substituting any of the found
$\varphi_k(x)$ into the previous equation
$R_{k-1,k-1}(x){\varphi}_{k-1}(x)+R_{k-1,k}(x){\varphi}_k(x)=0$, we
find that ${\varphi}_{k-1}$ also belongs to a finite-dimensional
space. Then analogously we define all the other
${\varphi}_{k-2}(x)$, \ldots ,${\varphi}_{1}(x)$, each of them
belong to a finite-dimensional space of solutions of some linear
ordinary differential equation. Therefore the dimension of the
kernel of the multiintegral is finite and equals the sum of the
orders of the differential operators on the principal diagonal of
the matrix $\widehat{R}$.

Suppose now that the matrix $\widehat R(x)$ is non-square. In any
case the algorithm \textbf{RMI} outputs an echelon matrix without
zero rows and, if it is non-square, the number of its rows is
smaller than the number of its columns. We show that the kernel of
such a matrix is infinite-dimensional. Take the first row of
$\widehat R(x)$ with a non-diagonal pivot and consider the previous
row, which has the first non-zero elements $R_{qq}\neq 0$,
$R_{q,q+1}$, \ldots , $R_{q,q+r}$ so that $R_{q+1,q+r+1}\neq 0$ is
the pivot of the next row, $r \geq 1$. Choose the following
functions $\varphi_{1}(x)$, \ldots ,$\varphi_k(x)$: set
$\varphi_{i}(x)\equiv 0$, if $i>q+r$; the other $\varphi_{q}(x)$,
\ldots ,$\varphi_{q+r}(x)$ should be taken from the space of
solutions of the differential equation $R_{qq}\varphi_q(x) + \ldots
+ R_{q,q+r}\varphi_{q+r}(x)=0$. This space is obviously
infinite-dimensional since the equation contains at least two
unknown functions $\varphi_i(x)$. Substituting the chosen
$\varphi_{q}(x)$, \ldots ,$\varphi_{q+r}(x)$ into the previous
equations of the system $\widehat R(x)\mathring{\varphi}(x)=0$, we
find the other $\varphi_{q-1}(x)$, \ldots ,$\varphi_{1}(x)$. So we
have found an infinite-dimensional subspace of the solution space of
the system $\widehat R(x)\mathring{\varphi}(x)=0$. The functional
dimension (the number of free functions of one variables other than
constants) of the complete kernel, as one can easily conclude from
the considerations above, is equal to the difference of the numbers
of the columns and rows of the matrix $\widehat R(x)$.
\end{proof}

\begin{lemma}\label{lemma.20a}
Let an operator matrix $\widehat{R}(x)$ of the size $k \times p$
with $k \leq p$ be given. Then $\widehat{R}(x)$ is non-invertible if
and only if it has a nontrivial kernel.
\end{lemma}
\begin{proof}
It is easy to see that nontriviality of the kernel implies that
$\widehat{R}(x)$ is non-invertible.

Conversely, let $\widehat{R}(x)$ be non-invertible. Reduce
$\widehat{R}(x)$ to echelon form. Multiplying the rows of the
resulting $\widehat{R}(x)$ with appropriate non-zero functions we
can without limitation of generality assume that its pivots of order
zero are equal to 1.

\emph{Case A}: let $k = p$, i.e.\ the matrix $\widehat{R}(x)$ is
square. Then among its pivots $R_{11}(x)$, \ldots , $R_{pp}(x)$ at
least one is not equal to 1 (otherwise $\widehat{R}(x)$ is
invertible). Take the first pivot different from 1. Let it be
$R_{mm}(x)$. The differential equation $R_{mm}(x)\varphi(x)=0$
always has a non-zero solution $\widetilde{\varphi}(x)$.
Consequently the kernel $\widehat{R}(x)$ contains a non-zero element
$\mathring{\varphi}(x)=(\varphi_{1}(x)$, \ldots ,$\varphi_p(x))^t$,
where

$\varphi_{m+1}(x)=\ldots = \varphi_{p}(x)=0$,

$\varphi_{m}(x)=\widetilde{\varphi}(x)$,

$\varphi_{m-1}(x)=- R_{m-1,m}(x)\varphi_{m}(x)$,

$\varphi_{m-2}(x)=-R_{m-2,m-1}(x)\varphi_{m-1}(x)-
R_{m-2,m}(x)\varphi_{m}(x)$,

 \ldots

$\varphi_{1}(x)=- \sum_{j=2}^mR_{1,j}(x)\varphi_{j}(x)$.

\emph{case B}: now we assume $\widehat{R}(x)$ to be non-square,
i.e.\ $k < p$. Take the first row of $\widehat{R}(x)$ with the pivot
$R_{mm}$ such that the next pivot is not diagonal (or take the last
row). We have the entries $R_{mm}\neq 0$, $R_{m,m+1}$, \ldots ,
$R_{m,m+r}$ of this row, such that  $R_{m+1,m+r+1}\neq 0$ is the
next pivot, $r \geq 1$ (or $m=p$ so we take all elements of the last
row). If the pivot $R_{mm}$ is not equal to 1, we can easily find a
nontrivial element of the kernel $\widehat{R}(x)$, using the
argumentation of the \emph{case~A}. If we have $r+1 \geq 2$ then we
can find an infinite-dimensional kernel of $\widehat{R}(x)$ using
the argumentation of the proof of Proposition~\ref{Utv.3.14}.
\end{proof}

\begin{theorem}\label{T-mi-ekv}
A multiintegral $\mathcal{J}$ is operator-reducible if and only if
$\mathcal{J}$ is kernel-reducible.
\end{theorem}
\begin{proof}
Applying to $\mathcal{J}=\overline{L}(x,y) \mathring{\varphi}(x)$
the algorithm \textbf{RMI}, we obtain the representation
$\overline{L}(x,y) = \overline{M}(x,y)\widehat R(x)$. If
$\mathcal{J}$ is kernel-reducible, then according to
Corollary~\ref{sleds.3.13} the matrix $\widehat R(x)$ has a
nontrivial kernel and is non-invertible. So $\mathcal{J}$ is
operator-reducible.

Conversely, let $\mathcal{J}$ be operator-reducible. i.e.\ we have
the corresponding decomposition $\overline{L}(x,y) =
\overline{M}(x,y)\widehat R(x)$, where $\widehat R(x)$ is a
non-invertible matrix, such that the number of its rows is smaller
or equal than the number of its columns. According to
Lemma~\ref{lemma.20a} the operator matrix $\widehat R(x)$ has a
nontrivial kernel. So $\mathcal{J}$ is kernel-reducible.
\end{proof}

From this Theorem and Proposition~\ref{utver-nered-M} we immediately
get

\begin{corollary}\label{3.6-}
The obtained in the algorithm~\textbf{RMI} multiintegral
$\mathcal{J}=\overline{M}(x,y)\mathring{\psi}(x)$ is
operator-reducible.
\end{corollary}

\begin{theorem}\label{teor.3.17}
Let two multiintegrals $\mathcal{J}_1=\big(M_1(x,y), \ldots ,
M_p(x,y)\big)$, $\mathcal{J}_2=\big(L_1(x,y), \ldots ,
L_k(x,y)\big)$ be given, and $\mathcal{J}_1$ is
operator-irreducible. Let the function stock generated by the
multiintegral $\mathcal{J}_2$ be contained (as a subset) in the
function stock generated by $\mathcal{J}_1$. Then
\begin{equation}\label{3.17.0}
\big(L_1(x,y), \ldots , L_k(x,y)\big) =\big(M_1(x,y), \ldots ,
M_p(x,y)\big)\widehat{R}(x)
\end{equation}
with some matrix $\widehat{R}(x)$ of operators whose coefficients
depend only on $x$. Moreover the rank of the matrix $\widehat{R}(x)$
coincides with the number $p$ of its rows if and only if the
function stocks generated by $\mathcal{J}_1$ and $\mathcal{J}_2$,
coincide.
\end{theorem}
\begin{proof}
Since we have $Z(\mathcal{J}_2) \subseteq Z(\mathcal{J}_1)$, then
for every column $\mathring{\varphi}(x)=(\varphi_{1}(x)$, \ldots
,$\varphi_k(x))^t$ one can find another column
$\mathring{\psi}(x)=(\psi_{1}(x)$, \ldots ,$\psi_p(x))^t$, such that
\begin{equation}\label{3.17.1}
\big(M_1(x,y), \ldots , M_p(x,y)\big)\mathring{\psi} =
\big(L_1(x,y), \ldots , L_k(x,y)\big)\mathring{\varphi}.
\end{equation}
Reproduce the procedures of the algorithm \textbf{RMI} for the
equation (\ref{3.17.1}) with the obvious modifications. Since
$M_1(x,y)$ is irreducible, one can choose the values $y_1$, \ldots ,
$y_s$ and operators $P_1(x)$, \ldots , $P_s(x)$ so that
$$
P_1(x)M_1(x,y_1) + \ldots +P_s(x)M_1(x,y_s) =1.
$$
Therefore from (\ref{3.17.1}) we can obtain
\begin{equation}\label{3.18.1}
\psi_1(x) + \sum_{i=1}^sP_i(x)M_2(x,y_i) \psi_2(x) + \ldots =
\sum_{i=1}^sP_i(x)L_1(x,y_i) \varphi_1(x) + \ldots
\end{equation}
Applying the operator $M_1(x,y)$ to the both sides of this equation
and subtracting the result from (\ref{3.17.1}) we arrive at an
analogue of the equality (\ref{3.11.15}) of the algorithm
\textbf{RMI}. Proceeding with the algorithm we finally obtain an
equality of the form
\begin{equation}\label{3.18.2}
Q(x,y)\psi_p(x) =\sum_{j=1}^kT_j(x)\varphi_j(x).
\end{equation}
Since by the assumption of the Theorem  $\mathcal{J}_1$ is
operator-irreducible, the obtained operator $Q(x,y)$ is also
operator-irreducible so choosing some appropriate points $\tilde
y_1$, \ldots , $\tilde y_t$ and operators $\tilde P_1(x)$, \ldots ,
$\tilde P_t(x)$ one has
$$
\tilde P_1(x)Q_1(x,\tilde y_1) + \ldots +\tilde P_t(x)Q_1(x,\tilde
y_t) =1.
$$
Then from (\ref{3.18.2}) we can easily deduce
\begin{equation}\label{3.19.1}
\psi_p(x) =\sum_{j=1}^k R_{pj}(x)\varphi_j(x).
\end{equation}
The found operators $R_{pj}(x)$ give us precisely the last row of
the matrix $\widehat{R}(x)$. Its previous rows can be readily found
from the obtained in the process of the algorithm equalities of the
form (\ref{3.18.1}). Substituting (\ref{3.19.1}) and the analogous
expressions for $\psi_{p-1}(x)$, \ldots , $\psi_{1}(x)$ into
(\ref{3.17.1}) we obtain the required formula (\ref{3.17.0}) since
the functions $\varphi_{i}(x)$ were chosen to be arbitrary.

From (\ref{3.17.0}) one concludes that the function stocks generated
by the multiintegrals in question coincide iff the system of linear
ordinary differential equations
\begin{equation}\label{3.20.1}
\widehat{R}(x)\mathring{\varphi}(x) =\mathring{\psi}(x)
\end{equation}
is solvable for every column $\mathring{\psi}(x)$. Reduce
$\widehat{R}$ to echelon form (this is equivalent to multiplication
of the both sides of (\ref{3.20.1}) on the left with some invertible
operator matrix $\widehat{S}(x)$). The system
$\widehat{S}(x)\widehat{R}(x)\mathring{\varphi}(x)
=\widehat{S}(x)\mathring{\psi}(x)$ is solvable for any
$\mathring{\psi}(x)$ if and only if the echelon matrix
$\widehat{S}(x)\widehat{R}(x)$ has no zero rows. This proves the
last statement of the Theorem.
\end{proof}

Note that contrary to the case of Euler integrals, a strict
inclusion $Z(\mathcal{J}_2) \subset Z(\mathcal{J}_1)$ is possible
for multiintegrals. The simplest case can be given by any
multiintegral of the form $\mathcal{J}_1=
\big(L_1(x,y),L_2(x,y)\big)({\varphi}_1(x),{\varphi}_2(x))^t$ where
both $L_1(x,y)$, $L_2(x,y)$ generate irreducible integrals with
different \emph{nonintersecting} function stocks
$Z(L_1{\varphi}_1)$, $Z(L_2{\varphi}_2)$. First of all we deduce
that $\mathcal{J}_1$ is kernel-irreducible and its function stock
coincides with the direct sum $Z(L_1{\varphi}_1) \oplus
Z(L_2{\varphi}_2)$. Take now $\mathcal{J}_2= L_1(x,y){\psi}_1(x)$.
We see that $Z(\mathcal{J}_2)=Z(L_2{\varphi}_2)$ is strictly
contained in $Z(\mathcal{J}_1)$.

This example motivates the introduction of the following notion
\begin{definition}
If an operator-irreducible multiintegral $\mathcal{J}=\big(L_1(x,y),
\ldots , L_k(x,y)\big)$ is given, then its \emph{submultiintegral}
is any multiintegral $\widetilde{\mathcal{J}}$ of the form
$$
\widetilde{\mathcal{J}} =\big(\widetilde{L}_1(x,y), \ldots ,
\widetilde{L}_m(x,y)\big)\mathring{\varphi}(x)=\big(L_1(x,y), \ldots
, L_k(x,y)\big)\widehat{N}(x)\mathring{\varphi}(x),
$$
where $\widehat{N}(x)$ is an arbitrary matrix of LODO with
coefficients depending on $x$. As obvious from the results proved
above, this is equivalent to the requirement
$Z(\widetilde{\mathcal{J}}) \subseteq Z(\mathcal{J})$.
\end{definition}

\textbf{Remark.} As we can see from the proof of
Theorem~\ref{teor.3.17}, there exists an algorithmic way to decide
for any two given multiintegrals $\mathcal{J}_1$ and
$\mathcal{J}_2$, wether the function stocks generated by them
coincide or are different, or wether we have a strict inclusion of
one stock into another and what is the dimension of their
intersection $Z(\mathcal{J}_2) \cap Z(\mathcal{J}_1)$.

To this end we find using the algorithm \textbf{RMI} the irreducible
representations of the given multiintegrals
$\overline{\mathcal{J}}_1=\big(L_1(x,y), \ldots ,
L_k(x,y)\big)\mathring{\varphi}(x)$,
$\overline{\mathcal{J}}_2=\big(M_1(x,y), \ldots ,
M_p(x,y)\big)\mathring{\psi}(x)$ (with the same function stocks).
Write the following formal equality
\begin{equation}\label{20a}
\big(M_1(x,y), \ldots , M_p(x,y)\big) \mathring{\psi}(x) =
\big(L_1(x,y), \ldots , L_k(x,y)\big) \mathring{\varphi}(x),
\end{equation}
where $\mathring{\psi}(x) =(\psi_{1}(x)$, \ldots ,$\psi_p(x))^t$,
$\mathring{\varphi}(x)=(\varphi_{1}(x)$, \ldots ,$\varphi_k(x))^t$.
Applying the procedures of the algorithm \textbf{RMI} to the
equality (\ref{20a}), as in the proof of Theorem~\ref{teor.3.17}, we
find a matrix $\widehat{R}(x)$ such that
\begin{equation}\label{20a1}
\mathring{\psi}(x) =\widehat{R}(x)\mathring{\varphi}(x).
\end{equation}
From (\ref{20a}) and (\ref{20a1}) one deduce that the intersection
of the function stocks of $\mathcal{J}_1$ and $\mathcal{J}_2$
coincides with the set $\overline{L}(x,y)\mathring{\varphi}(x)$,
where $\mathring{\varphi}(x)$ are solutions of the equations
$$
\overline{M}(x,y)\widehat{R}(x)\mathring{\varphi}(x)
=\overline{L}(x,y)\mathring{\varphi}(x).
$$
Thus if we have the operator identity
\begin{equation}\label{20a1a}
\overline{L}(x,y)=\overline{M}(x,y)\widehat{R}(x),
\end{equation}
then either $Z(\mathcal{J}_2) \subseteq Z(\mathcal{J}_1)$ so by
definition $\mathcal{J}_1$ is a proper submultiintegral of
$\mathcal{J}_2$ (in the case when the rank of $\widehat{R}$ is
smaller than the number of its rows $p$), or $Z(\mathcal{J}_2) =
Z(\mathcal{J}_1)$ (if  the rank of $\widehat{R}$ equals the number
of its rows $p$). If the equality (\ref{20a1a}) does not hold we
form a new multiintegral
$\mathcal{J}_3=\overline{K}(x,y)\mathring{\xi}(x)$ with
$\overline{K}(x,y)=\overline{L}(x,y)-\overline{M}(x,y)\widehat{R}(x)$.
Applying to $\mathcal{J}_3$ the algorithm \textbf{RMI}, we obtain
the representation
\begin{equation}\label{20a2}
\overline{K}(x,y)=\overline{N}(x,y)\widehat{R}_1(x).
\end{equation}
The intersection of the function stocks of $\mathcal{J}_1$ and
$\mathcal{J}_2$ coincides with the set
$\{\overline{L}(x,y)\mathring{\xi}(x)\, | \,\mathring{\xi}(x) \in
\textrm{Ker} \widehat{R}_1(x)\}$ so by Proposition~\ref{Utv.3.14} it
is finite-dimensional if and only if the matrix $\widehat{R}_1(x)$
is a square matrix. In the contrary case (when the number of its
columns $k$ is greater than the number of its rows $p_1$)
$Z(\mathcal{J}_2)$ and $Z(\mathcal{J}_1)$ has an
infinite-dimensional intersection with the functional dimension
equal to the difference $k-p_1$.

\begin{proposition}\label{Utv-9a-mi}
Among multiintegrals generating the same function stock, there is a
unique irreducible  $\mathcal{I}_0 = \overline{L}_0(x,y)\psi (x)$
(unique up  to transformations $L_0\mapsto L_0(x,y)\widehat R(x)$
with any invertible operator matrix $\widehat R(x)$). All other
multiintegrals with the same function stock have the form
$\mathcal{I} = \overline{L}_0(x,y)\widehat
S(x)\mathring{\varphi}(x)$, where $\widehat S(x)$ is an arbitrary
operator matrix with coefficients depending on $x$, the number of
the rows of this matrix is equal to its rank.
\end{proposition}

The proof of this Theorem repeats the proof of
Proposition~\ref{Utv-9a} and uses Theorem~\ref{teor.3.17} instead of
Lemma~\ref{lemma-1}.

Note that in the approach we used so far we did not allow operations
with the arbitrary functions $\varphi_{i}(x)$. The natural
admissible operations with the function sets
$\mathring{\varphi}(x)=(\varphi_{1}(x)$, \ldots ,$\varphi_k(x))^t$
are:
\begin{itemize}
  \item transposition of two functions $\varphi_{i} \leftrightarrow
  \varphi_{j}$;
  \item multiplication of one function with a given nonzero
  multiplier  $g(x)$;
  \item substitution $\varphi_{j} \mapsto \varphi_{j} + A(x) \varphi_{i}$,
  $i \neq j$, where $A(x)$ is a LODO with coefficients depending
  only on $x$.
\end{itemize}
These operations correspond in the representation (\ref{3.3}) to:
\begin{itemize}
\item transposition of the corresponding operators
 $L_i(x,y)\leftrightarrow L_j(x,y)$ of the operator row $\overline{L}(x,y)$ and
 the corresponding columns of the matrix $\widehat{R}(x)$;

  \item multiplication of the corresponding operator $L_i(x,y)$  and the
  $i$-th column of the matrix $\widehat{R}(x)$ on the right with $g(x)$;

  \item substitution $L_i \rightarrow L_i +L_jA(x)$ and addition to the  $i$-th
  column of
  $\widehat{R}(x)$ of the $j$-th column multiplied on the right with
  $A(x)$.
\end{itemize}

All these operations correspond to multiplication of the both sides
of the equality (\ref{3.3}) on the right with an invertible operator
matrix. It was proved in \cite{Jacobson} that carrying out
simultaneously the aforementioned operations with \emph{both} the
columns and the rows of the matrix $\widehat{R}(x)$ lets us to
reduce it to the form $\widehat{\widetilde{R}}(x)$ where
$\widetilde{R}_{ii}(x) = 1$ for $1\leq i\leq p-1$,
$\widetilde{R}_{pp}(x)$ is a LODO, and all the other elements of the
matrix $\widehat{R}(x)$ vanish.

From this we deduce that applying to the obtained in the algorithm
\textbf{RMI} decomposition $\overline{L}= \overline{M}  \widehat{R}$
the described above  operations with the columns and the rows we get
the following representation of any multiintegral:
$$
\big(\widetilde{L}_1(x,y),\ldots, \widetilde{L}_k(x,y)\big) =
\big(\widetilde{M}_1(x,y),\ldots,
\widetilde{M}_p(x,y)\big)\widehat{\widetilde{R}}(x),
$$
that is $\widetilde{L}_1 = \widetilde{M}_1$, \ldots,
$\widetilde{L}_{p-1} = \widetilde{M}_{p-1}$, $\widetilde{L}_p =
\widetilde{M}_p \widetilde{R}_{pp}$ and, if $k>p$,
$\widetilde{L}_{p+1} =0$, \ldots , $\widetilde{L}_{k} =0$.

Thus the following Proposition is proved:

\begin{theorem}\label{theor.3.26}
Using the admissible operations with the set
$\mathring{\varphi}(x)=(\varphi_{1}(x)$, \ldots ,$\varphi_k(x))^t$
and  dropping zero components, any multiintegral can be reduced to
the form
$$
\mathcal{J}=\big(\widetilde{L}_1(x,y),\ldots,
\widetilde{L}_{p-1}(x,y),
\widetilde{M}_p(x,y)\widetilde{R}(x)\big)\mathring{\widetilde{\varphi}}(x),
$$
where the row $\big(\widetilde{L}_1,\ldots, \widetilde{L}_{p-1},
\widetilde{M}_p)$ gives an irreducible integral.
\end{theorem}


\section*{Acknowledgements}

The author expresses her sincere gratitude to Prof. O.V.Kaptsov and
Yu.V.Shanko for fruitful discussion of the results.

\bibliographystyle{plain}

\begin{thebibliography}{99}

\bibitem{Jacobson}
{\it Jacobson N.}
\newblock {The theory of rings}.
\newblock  AMS,  1943.

\bibitem{2ZhS2001} {\it  Zhiber A.V., Sokolov V.V.},
\newblock Exactly integrable hyperbolic equations of Liouville type.
\newblock { Russian Math. Surveys}. 2001. V.~56, No~1. p.~61--101.

\bibitem{Zh-St}
 {\it Zhiber V., Startsev S.Ya.},
\newblock Integrals, Solutions, and Existence Problems for Laplace
Transformations of Linear Hyperbolic Systems.
\newblock Mathematical Notes, 2003,  V.~74,
No.~5--6, p.~803--811.

\bibitem{Lopat}
 {\it Lopatinskij Ya.B.},
\newblock {Linear differential operators}: 
Doctoral Thesis, 71~p. Baku, 1946. 
\newblock {reprinted in}: {\it Ya. B. Lopatinskij}.
{General theory of boundary problems}. Kiev, 1984.


\bibitem{St08}
 {\it  Startsev S.Ya.},
\newblock Cascade method of Laplace integration for linear hyperbolic
systems of equations.
\newblock Mathematical Notes, 2008, V.~83, No.~1--2,
P.~97--106.

\bibitem{TT07}
 {\it  Taimanov I.A., Tsarev S.P.},
\newblock Two-dimensional {S}chr{\"o}dinger operators with fast decaying
              potential and multidimensional {$L_2$}-kernel,
\newblock Russian Math. Surveys, 2007, v.~62, $N^o$~3, p.~217-218.




\bibitem{Tsarev99} {\it Tsarev S.P.},
\newblock On Darboux-Integrable Nonlinear Partial Differential Equations.
\newblock Proceedings of the Steklov Institute of Mathematics, 1999, v.~225,
p.~372--381.

\bibitem{Eul-III}
{\it Euler L.} Institutionum calculi integralis. V.~III, Ac. Sc.
Petropoli, 1770.

\bibitem{anderson}
{\it Anderson I.M., Kamran N.}
\newblock The Variational Bicomplex for Second Order Scalar
 Partial Differential Equations in the Plane.
\newblock {Duke Math. J.}, 1997. V.~87. N~2. P.~265--319.

\bibitem{athorne_Nimmo_moutard}
{\it Athorne C.,  Nimmo J.J.C.}
\newblock On the Moutard transformation for integrable partial differential
  equations.
\newblock { Inverse Problems}, 1991, v.~7(6), p.~809--826.


\bibitem{Athorne}
{\it Athorne C.}
\newblock A ${\bf Z}^2 \times {\bf R}^3$
Toda system.
\newblock {Phys. Lett. A}. 1995.
 v.~206, p.~162--166.

\bibitem{forsyth}
{\it Forsyth A.R.}
\newblock {Theory of differential equations}.
\newblock Part IV, vol.~VI. Cambridge, 1906.

\bibitem{Darboux2}
{\it Darboux G.}
\newblock {Le\c{c}ons sur la th\'eorie g\'en\'erale des surfaces et les
  applications g\'eom\'etriques du calcul infinit\'esimal}, T.~2.
\newblock Gauthier-Villars, 1889.

\bibitem{glt07} {\it Ganzha E.I., Loginov V.M., Tsarev S.P.}
Exact solutions of hyperbolic systems of kinetic equations.
Application to Verhulst model with random perturbation
Mathematics of Computation, 2008, v.~1, No~3, p.~459--472. \\
e-print {\tt http://www.arxiv.org/}, math.AP/0612793.

\bibitem{Gu}
{\it Goursat \'E.}
\newblock Le\c{c}ons sur l'int\'egration des \'equations
 aux d\'eriv\'ees partielles du seconde ordre a deux variables
 ind\'ependants. T.~2.
\newblock Paris: Hermann, 1898.

\bibitem{LeRoux}
{\it Le Roux J.}
\newblock Extensions de la m\'ethode de Laplace aux \'equations
 lin\'eaires aux deriv\'ees partielles d'ordre
 sup\'erieur au second.
\newblock {Bull. Soc. Math. France}. 1899.
 V.~27. P.~237--262. A digitized copy is obtainable from
{\tt http://www.numdam.org/}

\bibitem{Petren}
{\it Petr\'en L.}
\newblock { Extension de la m\'ethode de Laplace aux
 \'equations $\sum_{i=0}^{n-1}A_{1i}\frac{\partial^{i+1}z}{\partial x\partial y^i}
 + \sum_{i=0}^{n}A_{0i}\frac{\partial^{i}z}{\partial y^i} = 0$}.
\newblock Lund Univ. Arsskrift. 1911. Bd.~7. Nr.~3. p.~1--166.

\bibitem{pisati}
{\it Pisati, L.}
\newblock  {Sulla estensione del metodo di Laplace alle
equazioni differenziali lineari di ordine qualunque con due
variabili indipendenti.}
\newblock {Rend. Circ. Matem. Palermo}. 1905. t.~20. P.~344--374.

\bibitem{ts05}
{\it Tsarev S.P.}
\newblock {Generalized Laplace Transformations and
Integration of Hyperbolic Systems of Linear Partial Differential
Equations.}
\newblock {Proc. ISSAC'2005 (July 24--27, 2005, Beijing, China)}
ACM Press. 2005. P.~325--331; also e-print cs.SC/0501030 at {\tt
http://www.archiv.org/}.

\bibitem{ts06}
{\it Tsarev S.P.}
\newblock {On factorization and solution of multidimensional
linear partial differential equations.}
\newblock  in:  "COMPUTER ALGEBRA 2006. Latest Advances in Symbolic Algorithms",
Proc. Waterloo Workshop, Canada, 10--12 April 2006, World
Scientific,  2007. p. 181-192.  e-print {\tt
http://www.archiv.org/}, cs.SC/0609075.



\end{thebibliography}

\end{document}